\RequirePackage{fix-cm}
\documentclass{svjour3}                     % onecolumn (standard format)
\usepackage{cite}

\usepackage[cmex10]{amsmath}
\usepackage{array}
\usepackage{mdwmath}
\usepackage{mdwtab}
\usepackage{graphicx}
\usepackage{euscript}
\usepackage{amsfonts}

\usepackage{epstopdf} % When using eps figures

\usepackage{bigstrut}
\usepackage{tabularx}
\usepackage{cite}

\usepackage{amsmath}
\usepackage{amssymb}
\usepackage{xcolor}
\newcommand*{\ImagWidth}{\linewidth}

\smartqed  % flush right qed marks, e.g. at end of proof

\newtheorem{Theorem}{Theorem}

\newtheorem{Proposition}{Proposition}

\begin{document}

\title{On Modeling Coverage and Rate of Random Cellular Networks under Generic Channel Fading}

\author{Akram~Al-Hourani, Sithamparanathan~Kandeepan}

\institute{A. Al-Hourani, S. Kandeepan  \at
            School of Electrical and Computer Engineering, RMIT University, Melbourne, Australia.\\
            \email{akram.hourani@ieee.org, kandeepan@ieee.org}\\
            The final publication is available at Springer via http://dx.doi.org/10.1007/s11276-015-1114-x
}

%\author{Author 1, Author 2, and Author 3

%\thanks{Manuscript received XX-April-2015 revised XX-XXXX-2015.}
%\thanks{A. Al-Hourani and S. Kandeepan are with the School of Electrical and Computer Engineering, RMIT University, Melbourne, Australia. E-mail: akram.hourani@ieee.org, kandeepan@ieee.org}
%\thanks{E. Hossain is with the Department of Electrical and Computer Engineering, University of Manitoba, Winnipeg, Canada. E-mail: ekram.hossain@umanitoba.ca}
%}

\markboth{Coverage and Rate of Cellular Networks}%
{Shell \MakeLowercase{\textit{et al.}}: Bare Demo of IEEEtran.cls for Journals}

\maketitle

\begin{abstract}
In this paper we provide an analytic framework for computing the expected downlink coverage probability, and the associated data rate of cellular networks, where base stations are distributed in a random manner. The provided expressions are in computable integral forms that accommodate generic channel fading conditions. We develop these expressions by modelling the cellular interference using stochastic geometry analysis, then we employ them for comparing the coverage resulting from various channel fading conditions namely Rayleigh and Rician fading, in addition to the fading-less channel. Furthermore, we expand the work to accommodate the effects of random frequency reuse on the cellular coverage and rate. Monte-Carlo simulations are conducted to validate the theoretical analysis, where the results show a very close match.
\end{abstract}

\section{Introduction}
The vast deployment scale of cellular communication has made it as one of the most ubiquitously available piece of infrastructure. The notable expansion rate of cellular networks is referred to the accelerating demand generated by mobile users, where network operators are endeavouring to bridge the gap between traffic load and the available network capacity by deploying additional base stations (BS). It is anticipated that within the next 5 years a data-traffic growth of around 10 folds will take place in cellular networks alone \cite{Cite_Ericsson}.

The locations of the deployed base stations, are usually constrained by many factors such as economical, urban planning codes, and the availability of land/utility etc. these factors are very difficult to control and to predict, which leads to an increasing randomness in the BS locations, where the theoretical hexagonal model is no longer feasible \cite{Japan_Nakata,Andrews_Capacity}. Due to the increasing complexity of cellular network, designers and researchers utilize simulation tools for predicting network coverage and performance. Such approach is widely accepted in the industry, however it can not give an analytical insight of the influence contributed by the vast simulation parameters. Rather, it provides a detailed case-specific solution with neither tractability nor flexibility.

Analytical insight of network dynamics is an essential enabler for strategic planning and long-term economical modelling \cite{Old_Baccelli}. And in order to capture the increasing irregularity of the network deployment, stochastic geometry models \cite{Book_Haenggi,Book_Stoyan,Book_Kingman,Book_Baccelli_1,Book_Baccelli_2} are recently gaining a paramount interest for studying wireless cellular networks. Stochastic geometry allows the analytical understanding of the performance of modern cellular technologies such as cognitive radios \cite{Book_Kandeepan}, heterogeneous networks, fractional frequency reuse \cite{FFR_Novlan} and device to device communications \cite{Ekram_D2D, Akram_D2D}, in addition to the fundamental coverage and capacity of the cellular network \cite{Andrews_Capacity,Andrew_Uplink,Japan_Ginbre,Japan_Nakata}. The most popular assumption in the literature for the radio power fading is the Rayleigh channel, where the distribution function of the received power takes a simple exponential shape. This assumption allows tractability and enormously simplifies the computation of expressions. However, the desired link might favour a better performance than Rayleigh model that is usually considered as the worst case scenario fading \cite{Book_Molisch}.

Understanding this gap in the literature, we propose in this paper an analytical approach for studying the coverage and the data rate of cellular networks under generic channel fading conditions, accommodating not only the fast-fading behaviour of the channel but also the possible effects of shadowing variation. We first model the cellular interference in a random cellular network, and then we study the expected performance metrics as spatially averaged over the entire network. In addition we demonstrate the effects on the coverage probability resulting from different channel fading scenarios namely; (i) fading-less channel, (ii) Rayleigh channel and (iii) Rician channel. We verify our analytical approach using Monte-Carlo simulations by running repeated random network deployments and obtaining the spatial average of the signal-to-noise-plus-interference ratio (SINR). These results are used to validate the analytical calculations obtained using the integral forms. The contribution of the paper could be summarized in the following points:
\begin{itemize}
  \item It provides a generic formula (the \emph{coverage equation}) for calculating the expected service success probability (or the \emph{coverage probability}) in random cellular networks, under generic channel fading conditions.
  \item The \emph{coverage equation} is flexible to allow different fading models for the serving signal from one side and the interfering signals from the other side.
  \item The paper provides a practical method to compute the expected data rate in a random cellular network, without resorting to complicated simulations.
\end{itemize}

The rest of this paper is structured as the following: in Section \ref{Sec_Related_Work} we provide a literature background and identify the key studies that our work is based upon. In Section \ref{Sec_Model} we build the network model and illustrate the implemented channel models. The analytical study of the cellular interference is presented in Section \ref{Sec_Interference}, while in Section \ref{Sec_Coverage} we derive the \emph{coverage equation} that describes the probability of successful communication in a computable integral form, we employ this equation for studying three different radio channels in Section \ref{Sec_Scenarios}. Section \ref{Sec_Capacity} explains our approach in estimating the network data rate. A comparison with Monte-Carlo simulations is provided in Section \ref{Sec_Simulation}. Finally in Section \ref{Sec_Conclution} we draw our conclusion remarks and the prospective research paradigms.

\section{Related Work}\label{Sec_Related_Work}
Several recent introductory works are available on stochastic geometry in the context of wireless networks \cite{Book_Haenggi,Haenggi_Tutorial,ElSawy_Survey}. However, some of the earliest work on this regards dates back to the 1970-1990 such as \cite{Old_Kleinrock,Old_Baccelli,Old_Sousa}. Since then, several leaps have taken place, for example the work in \cite{Win_Interference} draws a mathematical framework for the statistical distribution of the interference generated by random wireless networks, where in our derivation of cellular interference we follow a similar approach, but taking into consideration the specific properties of cellular networks. % That is, no interferes could be closer to the user than the serving base station.
Other works related to interference can be found in \cite{Haenggi_Interference} addressing slotted ALOHA interference topic assuming a Rayleigh fading channel, while the authors in \cite{Cite_Baccelli_Aloha} address a general fading channel and obtain the optimum transmission probability in slotted ALOHA network. Also the work in \cite{6527333} addresses inter-user interference in an extensive experiment using both slotted and unslotted CSMA/CA. In the context of interference, the work in \cite{Andrew_Interference} addresses clustered interferers but it is distinguished by addressing the amplitude and phase of the interference where the interfering signals can interact constructively or detractively, while that the vast of the literature deals with the aggregated power of the interferers, that is the algebraic sum of the power of all interfering signals. In \cite{Zhang_Haenggi_Interference} the authors employ stochastic geometry analysis on studying intercell interference coordination (ICIC) by muting the transmission from $K$ number of neighbouring stations on specific resource blocks.

The tractability facilitated by the Poisson point process (PPP) attracts researches to represent the BS locations according to this process. However, other studies in this field capture the possible repulsion between base stations, utilizing determinantal point processes \cite{Japan_Ginbre, Japan_Nakata}. The accuracy of PPP is proven to increase when heavy shadowing conditions affect the network \cite{Blaszczyszyn_PPP}, making PPP a valid assumption in most of practical network deployment scenarios.

Applying stochastic geometry for studying cellular communication is an appealing approach for what it can yield of analytical estimations of the different attributes affecting such networks. For example the work in \cite{Andrews_Capacity} addresses the probability of coverage in cellular networks assuming a Rayleigh fading channel affecting the serving signal, while the work in \cite{Dhillon_KTier} extends the same approach for multi-tier heterogeneous cellular network, where base stations are implemented with different power and capacity levels. The work in \cite{DiRenzo} provides a mathematical framework to compute the expected cellular data rate without the need to obtain the coverage probability, the utilized method depends on the moment generating function of the interference. %Aside from studying the fundamental properties of a cellular networks, stochastic geometry is utilized to test intercell interference mitigation mechanisms for example using fractional frequency reuse \cite{FFR_Novlan}.

The main difference in our work presented here, is that the coverage and capacity estimations can be obtained for any stochastic fading channel model, with the freedom to select different stochastic processes for the serving signal and for the interfering ones.

\section{Network Model}\label{Sec_Model}
Achieving high accuracy in estimating network performance requires system level simulation, usually performed for a specific wireless technology that is implemented on a deterministic geometrical environment. Such simulations are quite useful in practical deployments of networks, for example, when wireless operators are deploying a new service, or when they are upgrading their infrastructure. However, an analytical tractable approach is preferred to get an insight of the different factors contributing to the network performance \cite{Andrews_Capacity, Cite_Baccelli_Aloha, Japan_Ginbre, Japan_Nakata, Akram_Energy_D2D, Akram_D2D}. These factors include (but are not limited to) base stations density, fading models, and resource allocation. Although tractable analysis can be achieved using simplified approximation of the deterministic hexagonal models \cite{Cite_Stuber_Book}, the drawback is that these models are very simplistic and might not reflect the true behaviour of the network. Practical cellular network deployments include vast randomness in the location of base stations, which cannot be captured by deterministic models. From this perspective, stochastic geometry is widely used in the research field to model the random location of base stations, where for simplicity it is common to assume a homogeneous PPP to model the BS locations. Accordingly, and in order to preserve the tractability we adopt the PPP network model assuming homogeneous BS intensity of value $\lambda$ (BS per unit area). The point process itself is denoted as $\Phi=\{X_n \in \mathbb{R}^2\}_{n \in \mathbb{N}}$ and is assumed to take place in the two dimensional Euclidean space $\mathbb{R}^2$. Mobile users are typically associated to the BS of highest received power, which is characterised with random behaviour (fast fading and shadowing). Thus, the cellular boundaries are rather probabilistic due to the random effect of fading. However, taking aside the fading effects, we can draw the average cellular boundary of each of the PPP points, simply by taking its Voronoi cell, defined as the region where all users are closer to the serving BS from any other BS \cite{Book_Haenggi}:
\begin{equation} \label{Eq_Voronoi_Cells}
V(X_n) \stackrel{\bigtriangleup}{=} \{u \in \mathbb{R}^2 : ||X_n-u|| \leq ||X_i-u|| \, \forall X_i \in \Phi \setminus \{X_n\} \},
\end{equation}
where $V(X_n)$ is the Voronoi cell of a base station $X_n$, and $\Phi$ is the set of base stations. This structure of the cellular system is called the Poisson Voronoi Tessellation (PVT) \cite{Book_Stoyan}. We depict a sample realization of a PVT layout in Fig. \ref{Fig_Layout}.
\begin{figure}
  \centering
  \includegraphics[width=\ImagWidth]{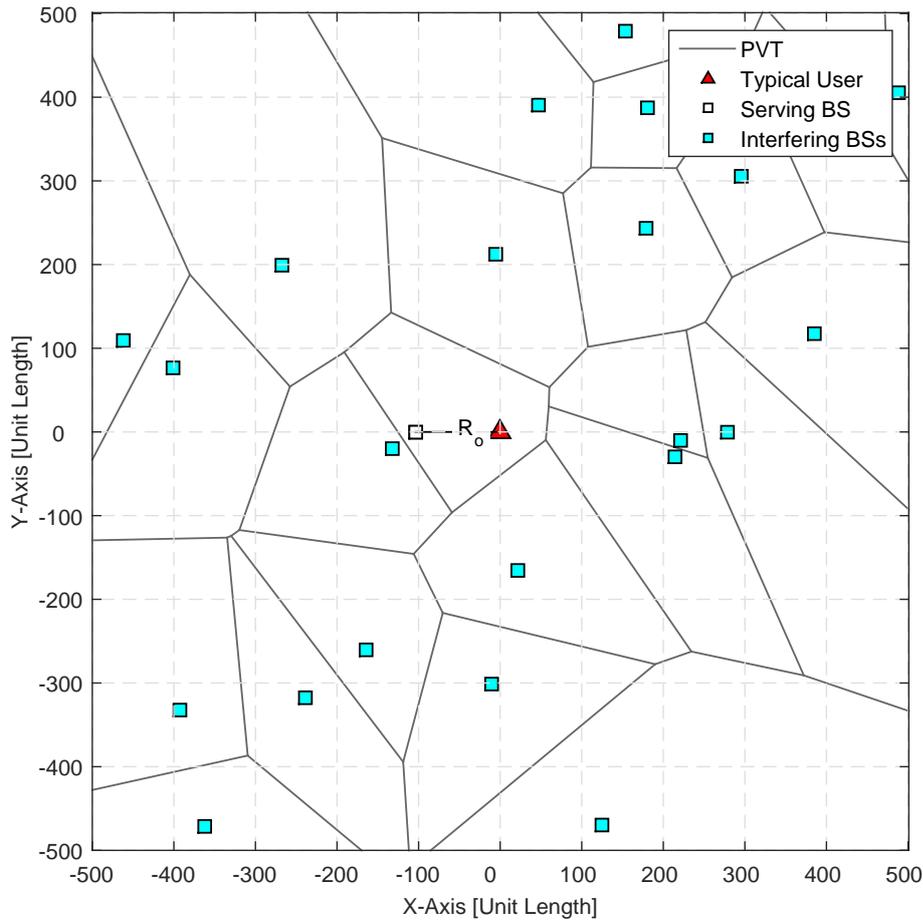}\\
  \caption{Cellular Poisson Voronoi tessellation, indicating the typical user located at the origin. }\label{Fig_Layout}
\end{figure}
%The reason that we can rely on PVT for determining the cell association is that the mean path-loss is monotonically increasing with respect to distance, so having a closer distance to a certain BS will result better \emph{average} received power, than any other BS.

The reason that we can rely on PVT for determining the cell association is that the mean path-loss is monotonically increasing with respect to distance, so having a closer distance to a certain BS will result is a better \emph{average} received power, than any other BS.

%PPP greatly facilitates the analytical derivation, and the tractability of the solution. Typical, a PPP is assumed to predict the lower bound of a real deployment \cite{Andrews_Capacity}\cite{Andrew_Uplink}\cite{ElSawy_Survey}. %however in \cite{Blaszczyszyn_PPP} Blaszczyszyn et al. show that in heavy log-normal shadowing conditions the network can be perceived from the user perspective as an infinite PPP.

We consider a \emph{homogeneous network} with all base stations having the same transmit power $P_o$, rather than a multi-tier heterogeneous network composed of base stations of variable power capabilities. No power control is accounted in our model, so that base stations are assumed to continuously transmit at a constant power level. Without loss of generality, we study a user located at the origin, where its statistical behaviour is typical for all other users in the network,  we call this mobile user as the \emph{typical user}, where we estimate its network performance for all possible spatial realizations of the random network. In other words, we implicitly assume a homogeneous distribution of network users, so that having the performance of the typical user will reflect the spatial average of all users in the network.

%We model the path-loss as log-distance relation given by $l(x)\stackrel{\bigtriangleup}{=}||x||^{-\alpha}$. The random behaviour of the slow fading is modelled as a random variable denoted as $\boldsymbol{g}$. The fast-fading is modelled as a random variable $\boldsymbol{h}$, that can represent any of the common fast-fading channel models, such as Rayleigh, Rician and m-Nakagami. Accordingly the resulting received power at a certain distance $x$ from a BS will have the following expression: $P_\mathrm{RX} = P_o  \boldsymbol{g}  \boldsymbol{h} l(x)$. %where $P_o$ represents the common power at which all base stations are transmitting.

\subsection{Channel Model}
Electromagnetic signals travelling between a BS and a receiver encounter power losses due to the propagation nature of the surrounding environment, resulting mainly from scattering, diffraction, reflection and absorption. These power losses are characterized with random behaviour and usually categorized into two distinct groups according to their rate of change, namely \emph{fast fading} and \emph{slow fading} or shadowing, where the random effect of these two categories is independent. Slow fading results mainly from the electromagnetic shadowing of obstacles, the random behaviour of the slow fading is modelled here as a random variable denoted as $\boldsymbol{g}$, where it is widely accepted to be considered to follow a log normal distribution according to the following:
\begin{equation}
\boldsymbol{g}= \exp\left(\sigma \boldsymbol{N}\right) : \boldsymbol{N} \sim \mathcal{N}(0,1)
\end{equation}
and $\sigma=\frac{\ln(10)}{10}\sigma_\mathrm{dB}$ represents the standard deviation, usually $\sigma_\mathrm{dB}$ is provided, representing the standard deviation in Decibel. The Gaussian distribution of zero mean and unity standard deviation is denoted as $\mathcal{N}(0,1)$.

On the other hand, the random interaction of multipath components at the receiver has a fast varying nature, causing rapid changes in the signal power. This effect is termed as the \emph{fast fading} and modelled here by a generic random variable $\boldsymbol{h}$, that can represent any of the common fast fading channel models, such as Rayleigh, Rician and m-Nakagami \cite{Book_Saakian}.

The mean loss due to the distance (the \emph{path-loss}) is modeled in a log-distance relation \cite{Book_Molisch}\cite{ITU_Channels}, so that at a location $x$ the mean path-loss between the origin and $x$ is given by the following:
\begin{equation}\label{Eq_Pathloss}
l(x)\stackrel{\bigtriangleup}{=}||x||^{-\alpha},
\end{equation}
where $||.||$ represents the Euclidean measure in $\mathbb{R}^2$, i.e. the distance between a source base station and the mobile station under study. Accordingly the resulting received power at a certain location will have the following expression:
\begin{align}
P_\mathrm{RX} = P_o . \boldsymbol{g} . \boldsymbol{h} . l(x),
\end{align}
where $P_o$ represents the common power at which all base stations are transmitting. Note that in this paper we represent random variables in \textbf{bold} for convenience and ease of interpretation.

\section{Modeling Cellular Interference}\label{Sec_Interference}
In our model, we assume that the BSs have a unity reuse factor. That is, for the typical mobile user, all BSs except the serving one are interfering with the downlink signal. Then the aggregated interference is given by
\begin{equation}
\boldsymbol{I} = \sum_{\Phi_\mathcal{I}} P_o  \boldsymbol{g_n} \boldsymbol{h_n} l(x) = \sum_{\Phi_\mathcal{I}} \boldsymbol{P_n} l(\boldsymbol{R_n}),
\end{equation}
where $\Phi_\mathcal{I}=\Phi\setminus \{X_o\}$ is the set of interferers, $X_o$ represents the serving BS, and $\{\boldsymbol{P_n}\}_{n \in \mathbb{N}^+}$ is a random variable vector having identical and independently distributed (i.i.d) elements, so that $\boldsymbol{P}=P_o\boldsymbol{g}\boldsymbol{h}$. The BSs' distances $\{\boldsymbol{R_n}\}_{n \in \mathbb{N}^+}$ constitute a random vector.

The illustration of a typical receiver located at the origin is shown in Fig. \ref{Fig_Typical}, where it is important to note that according to the assumed cellular association, all interfering BSs should be located outside the ball $b(o, \boldsymbol{R_o})$ of radius $\boldsymbol{R_o}$ and centred at the origin, where $\boldsymbol{R_o}$ is the distance to the serving BS, which is the contact distance to $\Phi$.

\begin{figure}
  \centering
  \includegraphics[width=\ImagWidth]{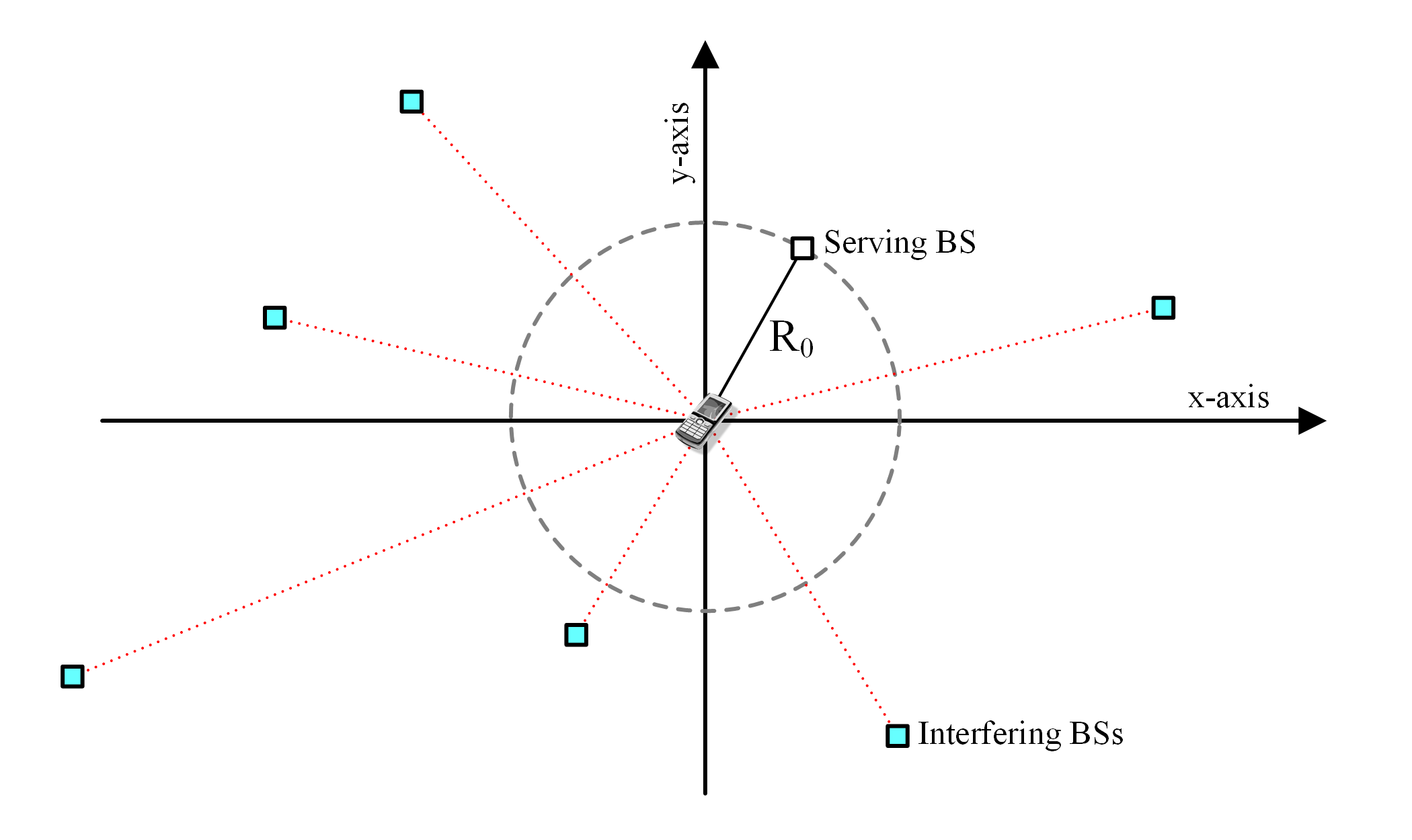}
  \caption{All interfering base stations are located outside the ball $b(o, \boldsymbol{R_o})$. }\label{Fig_Typical}
\end{figure}

The statistical distribution of the interference cannot be obtained for a generic case \cite{Win_Interference}; however, we can still deduce its Characteristic Function (CF) $\phi_{\boldsymbol{I}}(\omega)$, where the CF for a random variable $\boldsymbol{X}$ is defined as $\varphi_{\boldsymbol{X}}(\omega) = \mathbb{E} \left[  e^{\jmath\omega\boldsymbol{X}}  \right]$, where $\jmath=\sqrt{-1}$. If the characteristic function is identified, we can utilize the Gil-Pelaez's inversion theorem \cite{Inversion_Theorem} to compute the Cumulative Distribution Function (CDF) of $\boldsymbol{X}$ according to the following:
\begin{equation}\label{Eq_Inversion_Theorem}
F_{\boldsymbol{X} } (x) = \frac{1}{2} - \frac{1} {\pi} \int_0^\infty \frac{1} {\omega} \mathrm{Im}  \left[ \varphi_{\boldsymbol{X}}(\omega) \exp(-\jmath\omega x) \right] \mathrm{d}\omega.
\end{equation}

\begin{Proposition}
The characteristic function of the aggregated interference in a cellular network is given by:
\begin{equation}
\varphi_{\boldsymbol{I}}(\omega) =  \exp\left(     - 2\pi \lambda \beta \right), \text{ where}
\end{equation}
\begin{equation}\label{Eq_Beta}
\beta(\omega)=\int_{\boldsymbol{R_o}}^{\infty} \left[ 1- \varphi_{ \boldsymbol{P} } \left(\omega l(r) \right) \right] r\mathrm{d}r.
\end{equation}
\end{Proposition}

\begin{proof}
We start from the definition of the characteristic function of the interference, where the expectation should be performed over (i) the stochastic processes in $\boldsymbol{P}$ and (ii) over the geometrical stochastic process of $\Phi_\mathcal{I}$:

%\begingroup\makeatletter\def\f@size{8}\check@mathfonts
\begin{align}\label{Eq_Derivation_1}
&\varphi_{\boldsymbol{I}}(\omega) = \mathbb{E}  \left[ e^{\jmath\omega\boldsymbol{I}} \right] = \mathbb{E}_{\Phi_\mathcal{I}}  \mathbb{E}_{\boldsymbol{P}} \left[ \exp\left(\jmath\omega\sum_{\Phi_\mathcal{I}} \boldsymbol{P_n}l(\boldsymbol{R_n}) \right) \right] \nonumber \\
&= \mathbb{E}_{\Phi_\mathcal{I}}  \mathbb{E}_{\boldsymbol{P}} \left[ \prod_{\Phi_\mathcal{I}} \exp\left(\jmath\omega\boldsymbol{P_n}l(\boldsymbol{R_n}) \right) \right] \nonumber \\
&\stackrel{(a)}{=} \mathbb{E}_{\Phi_\mathcal{I}} \left[ \prod_{\Phi_\mathcal{I}}  \mathbb{E}_{\boldsymbol{P}} \left[ \exp\left(\jmath\omega\boldsymbol{P_n}l(\boldsymbol{R_n}) \right) \right] \right] \stackrel{(b)}{=} \mathbb{E}_{\Phi_\mathcal{I}} \left[ \prod_{\Phi_\mathcal{I}}   \varphi_{ \boldsymbol{P} } \left(\omega l(\boldsymbol{R_n}) \right)  \right],
\end{align}
%\endgroup
where (a) follows from the fact that the combined channel process $\boldsymbol{P}=P_o \boldsymbol{g} \boldsymbol{h}$ is independent of the geometrical process and (b) follows directly from the definition of the CF.
Now we can apply the probability generating functional of a homogeneous PPP on $\mathbb{R}^2$ \cite{Book_Haenggi}, where for a function $f(x)$ the following relation is satisfied:
\begin{equation}\label{Eq_PGFL_PPP}
\mathbb{E}\left[ \prod_{\Phi_\mathcal{I}} f(x) \right] = \exp\left(   - 2\pi \lambda \int_{\mathcal{I}} [1-f(x)] r\mathrm{d}r  \right).
\end{equation}
The integration variable $r \in \mathcal{I}$ is the distance range where the active interferers are located, that is $\mathcal{I} = (\boldsymbol{R_o},\infty)$. Accordingly, we can write (\ref{Eq_Derivation_1}) as the following:
\begin{equation}
\varphi_{\boldsymbol{I}}(\omega) =  \exp\left(     - 2\pi \lambda \int_{\boldsymbol{R_o}}^{\infty} \left[ 1- \varphi_{ \boldsymbol{P} } \left(\omega l(r) \right) \right] r\mathrm{d}r \right).
\end{equation}
Hence, Proposition 1 is proved.
\end{proof}
For a log-distance mean path-loss model, (\ref{Eq_Beta}) yields:
\begin{align}
\beta(\omega)&=\int_{\boldsymbol{R_o}}^{\infty} \left[ 1- \varphi_{ \boldsymbol{P} } \left(\omega r^{-\alpha}  \right) r\mathrm{d}r  \right], \text{ noting that}
\end{align}
\begin{align}
\varphi_{ \boldsymbol{P} } \left(\omega r^{-\alpha}  \right) &= \mathbb{E}_{\boldsymbol{g},\boldsymbol{h}} \left[ \exp\left(\jmath \omega r^{-\alpha} P_o  \boldsymbol{g}\boldsymbol{h} \right)\right].
\end{align}
Accordingly, we can rewrite $\beta$ as the following:
\begingroup\makeatletter\def\f@size{8.5}\check@mathfonts
\begin{align} \label{Eq_Beta_2}
\beta(\omega) &=\mathbb{E}_{\boldsymbol{g},\boldsymbol{h}} \left[ \int_{\boldsymbol{R_o}}^{\infty} \left[ 1- \exp \left(   \jmath \omega r^{-\alpha} P_o \boldsymbol{g}\boldsymbol{h} \right) \right]
\right] r\mathrm{d}r   &&\nonumber \\
              &=\mathbb{E}_{\boldsymbol{g},\boldsymbol{h}} \left[-\frac{\boldsymbol{R_o}^2}{2}+\frac{(-\jmath P_o \boldsymbol{g}\boldsymbol{h}\omega )^{2/\alpha } }{\alpha} \left[ \vphantom{\Gamma \left(\frac{-2}{\alpha }\right)} %\right. \right. \nonumber \\& \qquad \left. \left.
               \Gamma \left(\frac{-2}{\alpha },-\jmath P_o \boldsymbol{g}\boldsymbol{h}\boldsymbol{R_o}^{-\alpha } \omega \right)-\Gamma \left(\frac{-2}{\alpha }\right) \right] \right],
\end{align}
\endgroup
where $\Gamma(.)$ and $\Gamma(.,.)$ are the Gamma and the incomplete Gamma functions respectively.

\section{Modeling Coverage Probability}\label{Sec_Coverage}
Signal to Interference and Noise Ratio (SINR) is an important measure that can determine the link throughput and the availability of the wireless service. The SINR represents the strength of the target signal compared to the counterpart interferers' combined power plus the thermal noise generated inside the receiver's electronics. The latter can be represented as an Additive White Gaussian Noise (AWGN). The SINR is expressed as:    $\mathrm{SINR} = \frac {\boldsymbol{S}} {\boldsymbol{I}+W}$, where $\boldsymbol{S} = P_o \boldsymbol{h_o} \boldsymbol{g_o}l(\boldsymbol{R_o})$ is the desired signal which carries the needed information from the serving BS, $\boldsymbol{I}$ is the aggregate interference power, and $W$ is the AWGN noise power. The random variables $\boldsymbol{h_o}$ and $\boldsymbol{g_o}$ model the fast and slow fading respectively of the serving BS channel.

\begin{Theorem}
The probability that a receiver to be covered by a certain level of cellular wireless service is given by:
\begin{equation}\label{Eq_Main}
p_c=\mathbb{E}_{\boldsymbol{h_o,g_o}} \left[\int_{r>0} F_{\boldsymbol{I}} \left( \frac{P_o \boldsymbol{h_o}\boldsymbol{g_o}}{Tr^\alpha}-W \right) f_{\boldsymbol{R_o}}(r) \mathrm{d}r \right].
\end{equation}
\end{Theorem}

\begin{proof}

The outage of wireless service occurs when the SINR level at a receiver falls below a threshold $T$. Accordingly, we can express the link-level success probability at a given distance $R_o$ and a given $h_o$, $g_o$ as:
%\begingroup\makeatletter\def\f@size{7.5}\check@mathfonts
\begin{align}\label{Eq_Ps}
p_L  =\mathbb{P}[\mathrm{SINR} \geq T | h_o, g_o,R_o] =\mathbb{P}[\boldsymbol{I} \le \frac{S}{T}-W] =F_{\boldsymbol{I}}\left( \frac{S}{T}-W\right),
\end{align}
%\endgroup
where $F_{\boldsymbol{I}}$ is the cumulative distribution function (CDF) of the interference $\boldsymbol{I}$. Thus, the success probability is found by averaging over $\boldsymbol{R_o}$, $\boldsymbol{h_o}$, and $\boldsymbol{g_o}$ as:
\begin{equation}
p_c = \mathbb{E}_{\boldsymbol{R_o}, \boldsymbol{h_o}, \boldsymbol{g_o}} \left[p_L\right] = \mathbb{E}_{\boldsymbol{h_o}, \boldsymbol{g_o}} \left[ \mathbb{E}_{\boldsymbol{R_o}} \left[p_L\right] \right].
\end{equation}
Accordingly, we can obtain the result in (\ref{Eq_Main}) by applying the expectation rule over the contact distance $\boldsymbol{R_o}$, having a probability density function of $f_{\boldsymbol{R_o}}(r)=2\lambda \pi r \exp\left( -\lambda \pi r^2\right)$ in a PPP cellular network. We should note that $F_{\boldsymbol{I}}$ is computed using Gil-Pelaez's inversion formula given in (\ref{Eq_Inversion_Theorem}).
\end{proof}

We call (\ref{Eq_Main}) the \emph{coverage equation}, constituting the main result of this work allowing the evaluation of the averaged network-level success probability. The various dynamics affecting the cellular service success probability are visualized in Fig. \ref{Fig_Coverage_Equation}, namely (i) the base station density, (ii) the common base stations' power, (iii) the path-loss model, (iv) slow fading model, (v) fast fading model, (vi) noise level, and finally (vii) the target SINR threshold.

\begin{figure}
  \centering
  \includegraphics[width=3.00in]{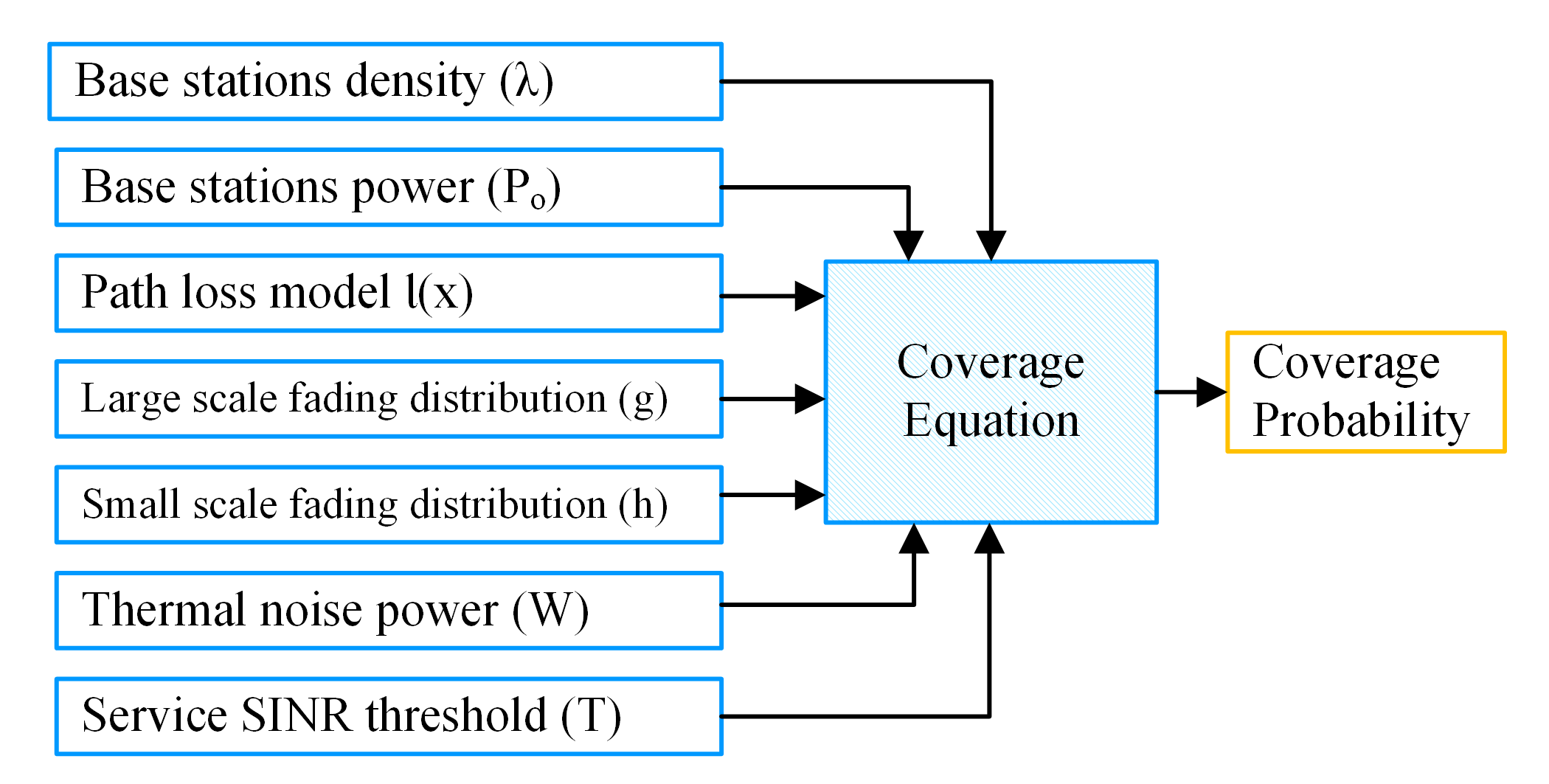}
  \caption{The different dynamics affecting the coverage probability. }\label{Fig_Coverage_Equation}
\end{figure}

\section{Network Performance Analysis}\label{Sec_Scenarios}
In this section, we first experiment a channel affected by a path-loss only, then we account for the fast-fading impairing both the serving BS signal and the interfering signals.

\subsection{Path-Loss Only Scenario}
By taking the effect of the path-loss only (i.e. $h=1$ and $g=1$) we can produce an initial understanding of the cellular coverage dynamics, representing the upper bound of the channel performance. According to (\ref{Eq_Beta_2}), $\beta$ can be reduced to:
\begin{flalign}
\beta(\omega) =-\frac{\boldsymbol{R_o}^2}{2}+\frac{(-\jmath P_o \omega )^{2/\alpha } }{\alpha} \left[ \vphantom{\Gamma \left(\frac{-2}{\alpha }\right)} \Gamma \left(\frac{-2}{\alpha },-\jmath P_o \boldsymbol{R_o}^{-\alpha } \omega \right)-\Gamma \left(\frac{-2}{\alpha }\right) \right].
\end{flalign}

We substitute different values of the path-loss exponent $\alpha$, assuming here that the cellular network is interference limited, so that the noise can be neglected (i.e., $W \to 0$). The results are counter-intuitive, as indicated in Fig. \ref{Fig_PathLoss_Scenario}, a better coverage in a random cellular network is achieved for higher values of the path-loss exponent, indicating that a heavier path-loss environment affects the aggregated interference more strongly than affecting the serving signal power.
\begin{figure}
  \centering
  \includegraphics[width=\ImagWidth]{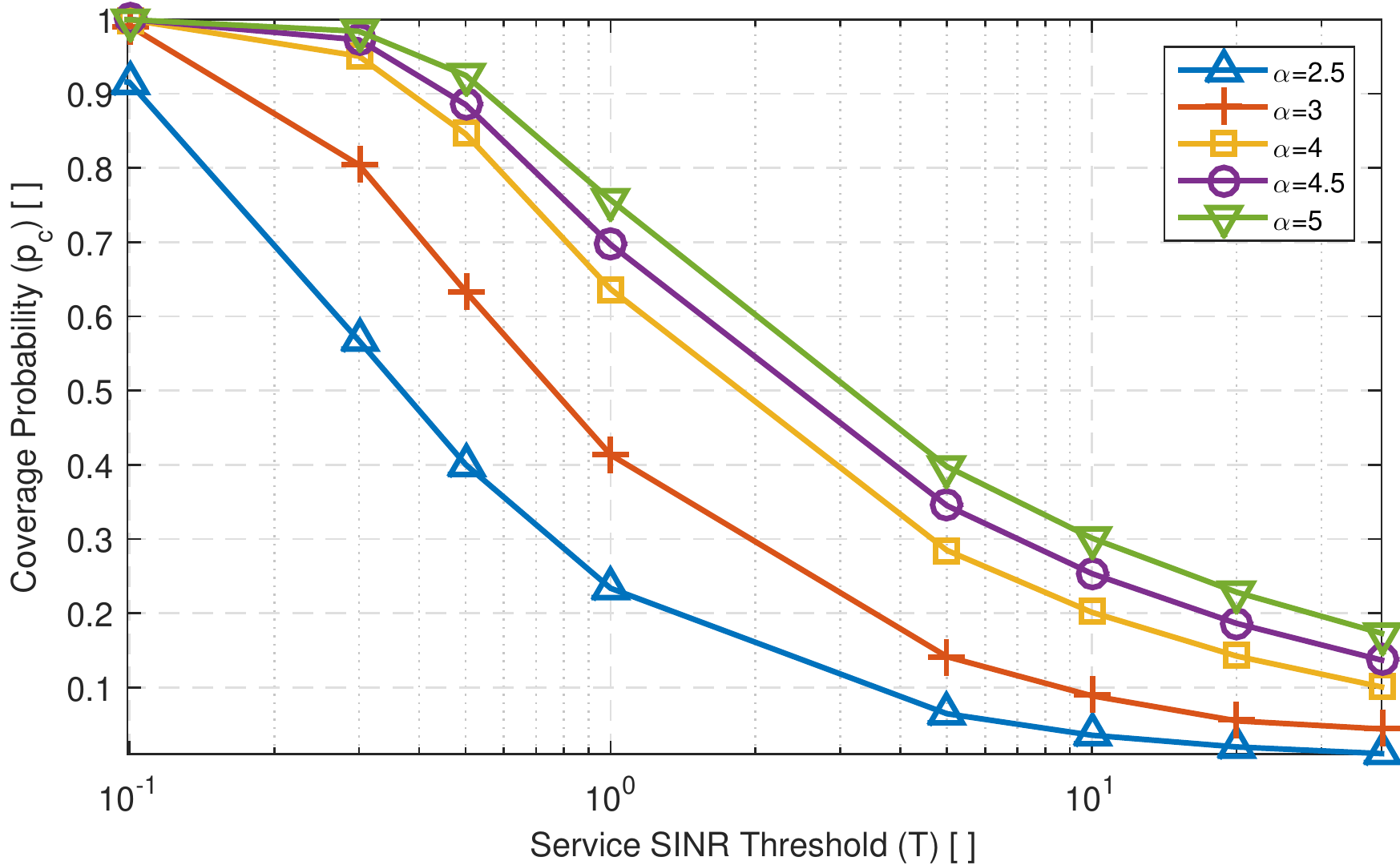}
  \caption{The coverage probability, when adopting log-distance path-loss channel model without considering fading effects.}\label{Fig_PathLoss_Scenario}
\end{figure}
Secondly, we compare the coverage probability for different intensities of BSs assuming a constant path-loss exponent. The results indicate that the BS intensity has insignificant effect on the coverage probability when the network is interference limited. Note that  the same observation was reported in \cite{Andrews_Capacity} and \cite{Japan_Ginbre} but for a Rayleigh fading channel.

\subsection{Rayleigh Fading Scenario} \label{Subsec_Scenario2}
In this scenario, we assume that both the serving and the interfering signals are impaired with Rayleigh fading. Namely, an exponential distribution random variable with a unity mean, where the probability density function (PDF) of $\boldsymbol{h}$ is given by $f_{\boldsymbol{h}}(x)= \exp(-x)$. Applying this to (\ref{Eq_Beta_2}), and by neglecting the shadowing variations (i.e. $g=1$), the result can be reduced to the following form:
\begingroup\makeatletter\def\f@size{8}\check@mathfonts
\begin{align}\label{Eq_Beta_Rayliegh}
\beta(\omega) &= -\frac{\boldsymbol{R_o}^2}{2} +  \frac{\pi}{\alpha} (-\jmath P_o \omega )^{2/\alpha }\csc \left(\frac{2 \pi }{\alpha }\right) + \jmath\frac{\boldsymbol{R_o}^{2 + \alpha}}{ (\alpha +2) P_o \omega} \, _2F_1\left(1,\frac{\alpha +2}{\alpha },\frac{2}{\alpha }+2,-\frac{\jmath \boldsymbol{R_o}^{\alpha }}{P_o \omega }\right),
\end{align}
\endgroup
where $_2F_1(.,.,.,.)$ is the hypergeometric function.

We perform numerical integration to calculate the coverage probability as per (\ref{Eq_Main}) for 5 different path-loss exponent values. We observe that higher path-loss exponent values have a favourable effect on the service success probability. However, Rayleigh scenario gives lower success probability than fading-less scenario, even though both the interference and the serving signals are affected by the same fading behaviour.
{\em We stress the point that the difference between the framework presented in this paper and the one in \cite{Andrews_Capacity} is the flexibility provided in choosing the fading model of the serving channel, so it is not limited to Rayleigh only}. This comes at the cost of more complex integral computation. We verify the results of the coverage equation with the ones obtained in \cite{Andrews_Capacity}, and depict the comparison in Fig. \ref{Fig_Rayleigh_Scenario}.

\begin{figure}
  \centering
  \includegraphics[width=\ImagWidth]{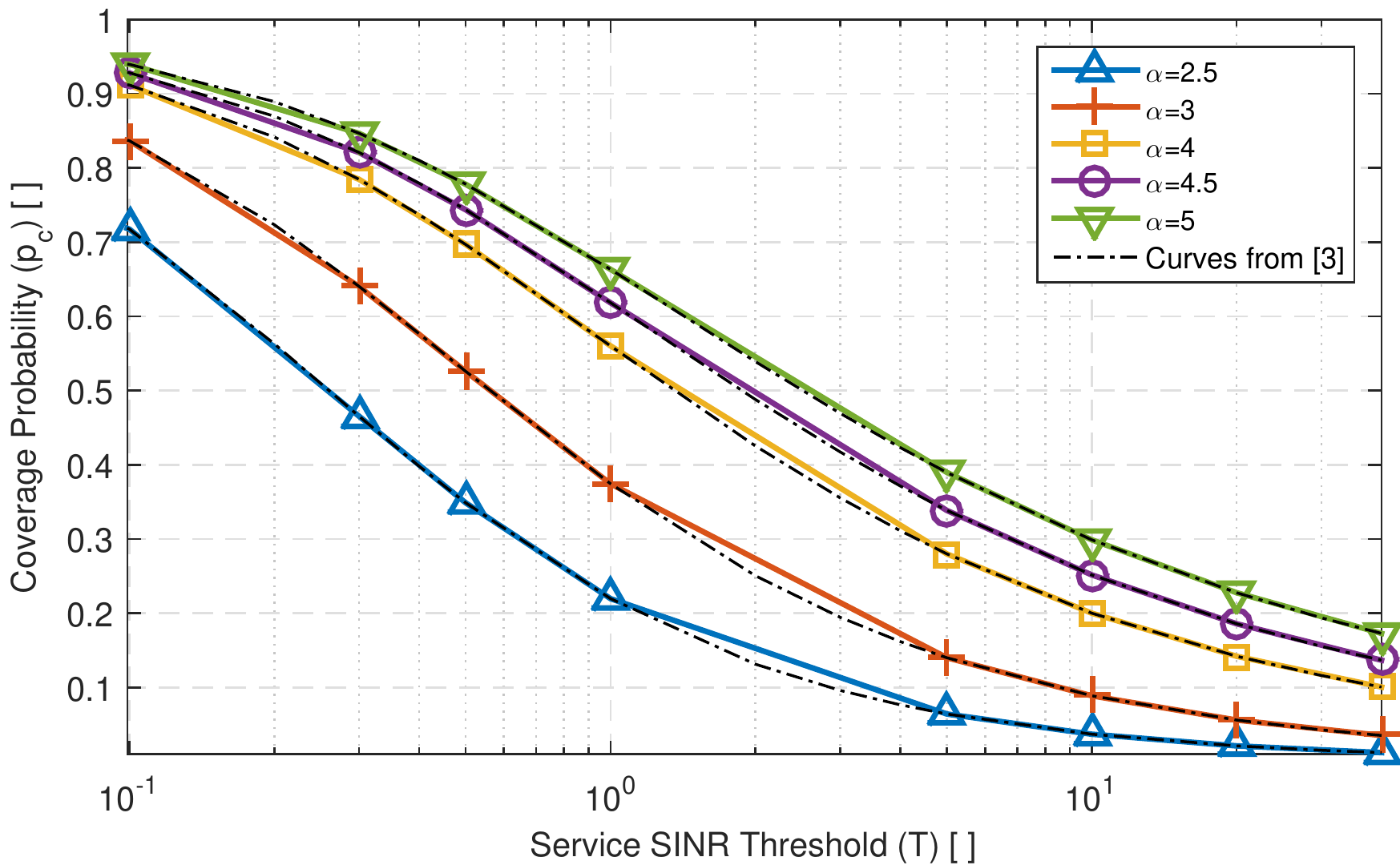}
  \caption{The coverage probability, for log-distance path-loss and Rayleigh fast-fading, showing a comparison with the results of \cite{Andrews_Capacity}.}\label{Fig_Rayleigh_Scenario}
\end{figure}

\subsection{Rayleigh Interferers with Rician Serving Signal}
Rician distributed fading can represent a wireless channel with more flexibility by tuning the $K$ factor which represents the ratio of the LoS  power to the sum of the powers from the defused multipath components. The probability distribution function describing a Rician fading channel gain is given as:
\begin{equation}
f_{\boldsymbol{h}}(x)=(K+1) e^{-x (K+1) - K } I_o \left(\sqrt{4x K (K+1)}\right),
\end{equation}
where $I_o(.)$ is the modified Bessel function of the first kind. In this scenario we assume that the serving signal follows a Rician distribution, while the interfering signals follow a Rayleigh distribution, thus $\beta(\omega)$ follows (\ref{Eq_Beta_Rayliegh}). The only difference will be in calculating the expectancy in the coverage equation (\ref{Eq_Main}), that is over a Rician distributed $\boldsymbol{h_o}$. The results are illustrated in Fig. \ref{Fig_Rician_Rayleigh_Scenario}, observing that the coverage probability is more sensitive to the distribution of the serving signal for lower SINR thresholds, a case which represents the edge users of the cell.

\begin{figure}
  \centering
  \includegraphics[width=\ImagWidth]{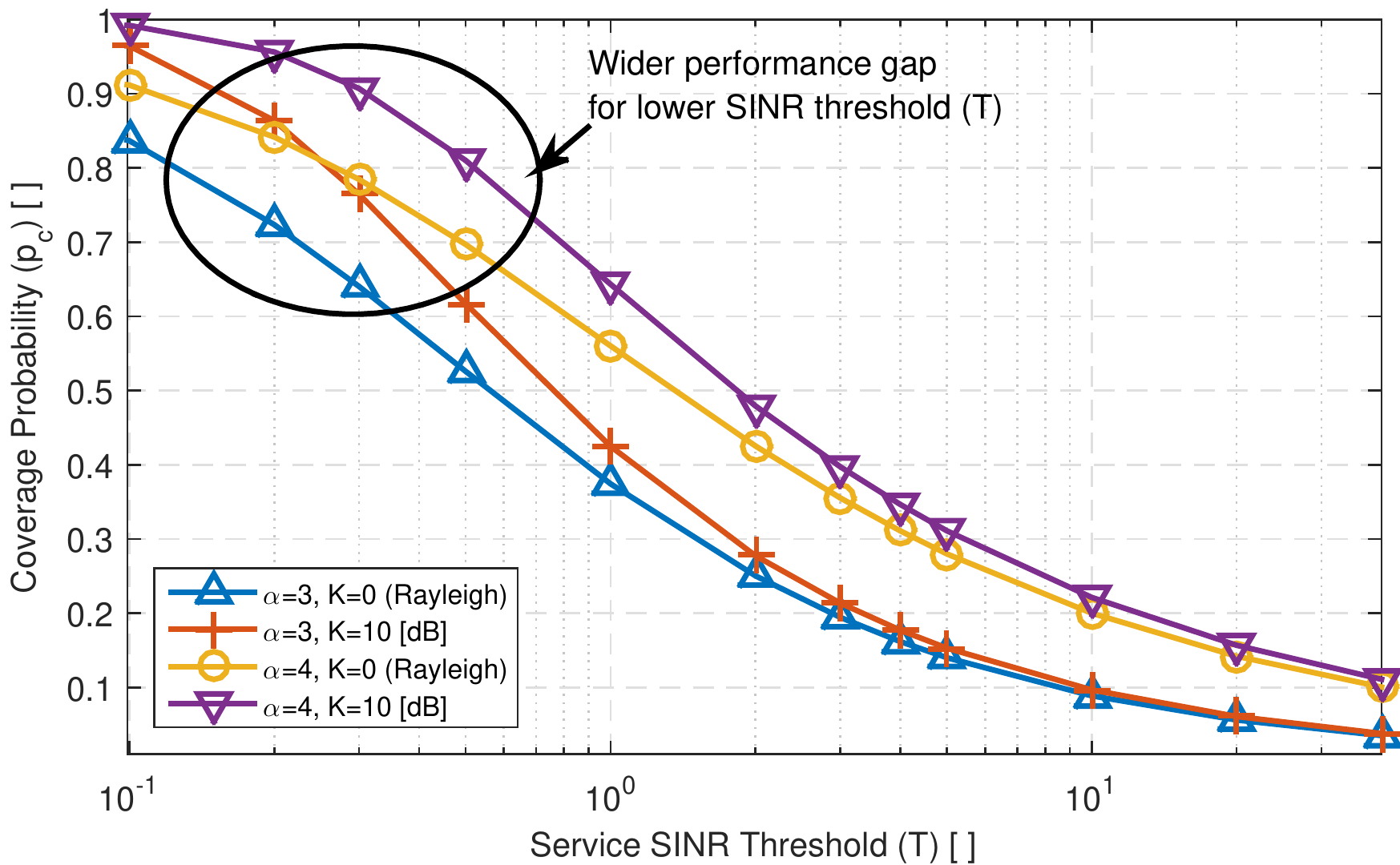}
  \caption{The coverage probability, for Rayleigh interfering signals and Rician serving signal.}\label{Fig_Rician_Rayleigh_Scenario}
\end{figure}

\subsection{Frequency Reuse} \label{Sec_FFR}
It can be clearly noticed how low is the SINR performance in the previous illustrated scenarios, also we note that increasing the density of base stations does not lead to a change in the coverage performance when the network is interference limited. Accordingly, resource management techniques should be applied to mitigate the co-channel interference between cells \cite{Zhang_Haenggi_Interference}. The spatial reuse of radio resource has always been the essence of cellular communication. However, there is a trade off between coverage performance and the spatial spectral efficiency. For example, when applying a frequency reuse scheme, the spectral efficiency will drop since the available spectrum for each cell will reduce. However a frequency reuse scheme will lessen the interference and boost the SINR performance.

In order to get a better insight of the expected network performance and how it is affected by radio resource coordination, we adopt the \emph{random frequency reuse scheme} due to its simplicity and tractability. In this scheme, base stations can choose from some $\Delta$ available radio frequencies. Thus each BS has a probability of a certain frequency assignment equal to $\frac{1}{\Delta}$. In this case, the co-channel interference will only be received from base stations utilizing the same frequency. The layout of the PVT will appear similar to Fig. \ref{Fig_Layout_FR}, where co-channel base stations are coloured the same. It is obvious that this channel assignment is not optimal, since co-channel cells are allowed to be mutual neighbours. But as mentioned before, the random frequency reuse greatly facilitates the mathematical analysis.
\begin{figure}
  \centering
  \includegraphics[width=\ImagWidth]{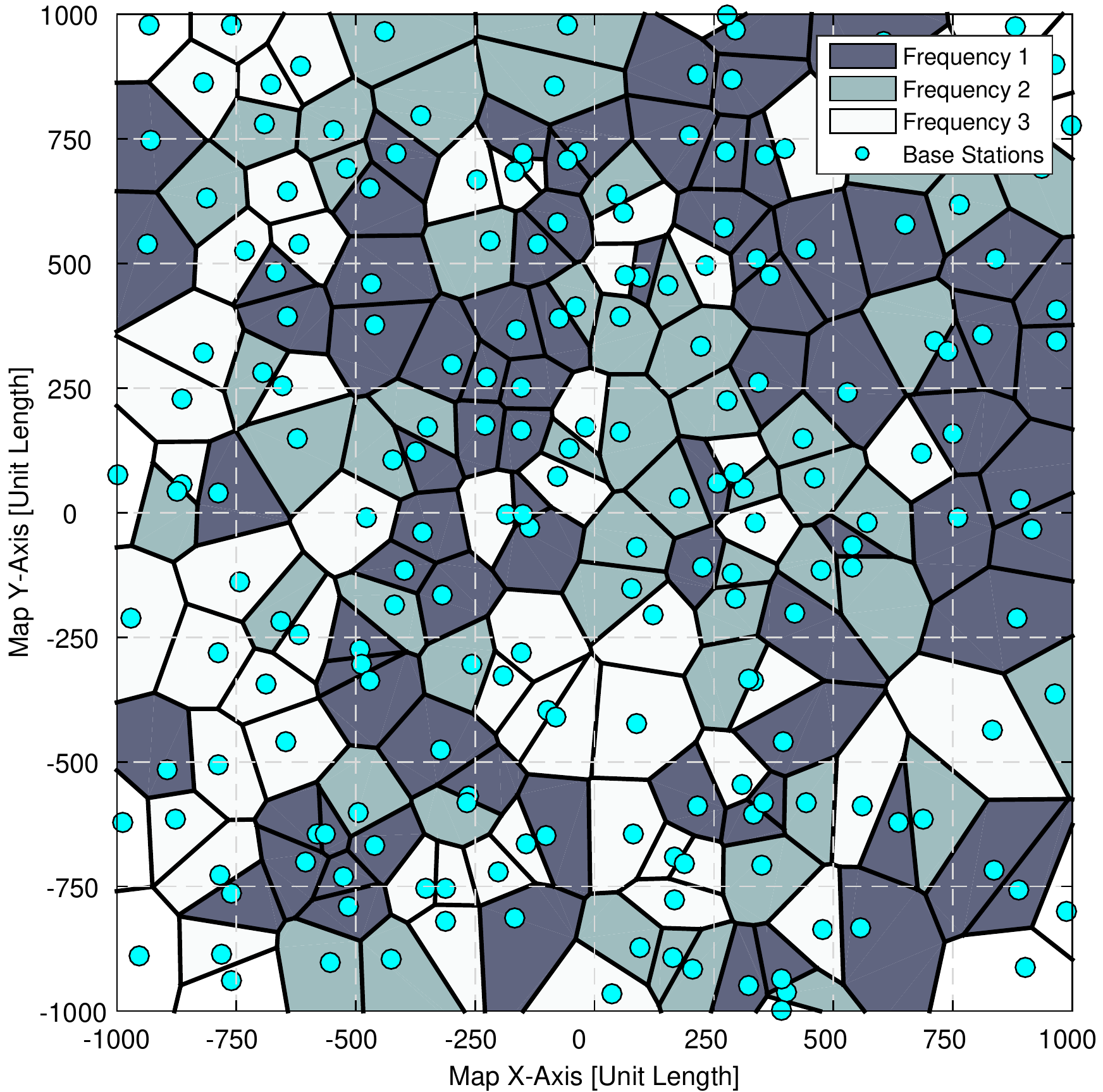}
  \caption{The layout of a random frequency reuse scheme in a Poisson Voronoi tessellation cellular network.}\label{Fig_Layout_FR}
\end{figure}

In Fig. \ref{Fig_FR} we plot the simulation results of the random frequency reuse scheme, showing the probability of the service success $p_c$ at an SINR threshold of $T=10$dB verses a range of reuse factor $\Delta$. An interesting observation is that the channel fading effect becomes more obvious for higher reuse factor, when comparing the three scenarios as explained in the subsections of Section \ref{Sec_Coverage}.
Simulation procedure will be explained in detail in Section \ref{Sec_Simulation}.

\begin{figure}
  \centering
  \includegraphics[width=\ImagWidth]{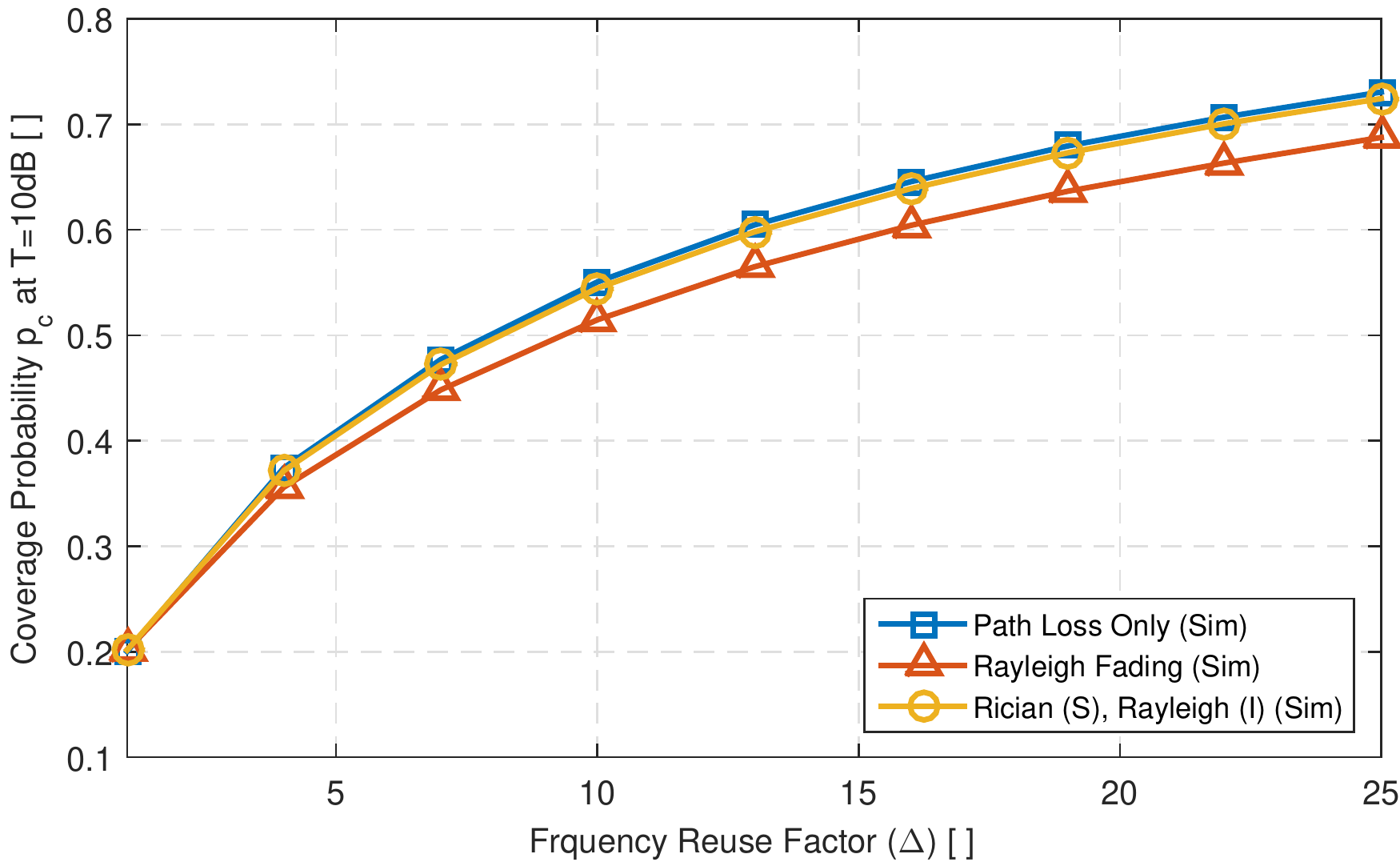}
  \caption{The effect of frequency reuse on the service success probability at an SINR threshold of $T=10$dB. where (S) refers to the serving signal and (I) refers to the interfering signals.}\label{Fig_FR}
\end{figure}

To analytically study the effect of random frequency reuse, we examine the new distribution of the interferers, noting that it follows a new PPP, since the independent thinning of a PPP will also yield a PPP \cite{Book_Baccelli_1} with a new intensity equal to $\frac{\lambda}{\Delta}$. In this case the coverage equation (\ref{Eq_Main}) resulted from Theorem 1 still holds, except that the aggregated interference follows a different stochastic distribution; the new CDF of the interference is denoted $\widehat{F_{\boldsymbol{I}}}$ so that its characteristic function is given by the following:
\begin{equation}
\widehat{\varphi_{\boldsymbol{I}}}(\omega) =  \exp\left(     - 2\pi \frac{\lambda}{\Delta} \beta \right),
\end{equation}
resulting in a reduced interference effect, thus a better coverage success probability.% as indicated in Fig. \ref{Fig_FR} that illustrates the comparison between the performance of the three scenarios stated earlier. Noticing that the performance gap between different channel models tends to expand for higher values of $\Delta$, that is having a better serving channel conditions are more effective when adopting radio resource coordination.

\section{Network Rate} \label{Sec_Capacity}
A practical system would perform less than the maximum limit set by Shannon capacity theorem. Accordingly we utilize a practical method \cite{NSN_LTE_Capacity} to estimate the user's throughput (data rate) per/Hz following the  expression: $\rho = \ln \left( 1+\frac{\boldsymbol{\mathrm{SINR}}}{\mathrm{SINR}_o} \right)$, where $\mathrm{SINR}_o$ is a system specific parameter considered here as a constant, representing the gap between Shannon limit and the achievable rate of the system. The estimated network rate can be calculated as:%is possible through numerical differentiation of the probability density function of the SINR, since the expected value of the network level rate can be calculated as:
%\begingroup\makeatletter\def\f@size{7}\check@mathfonts
\begin{equation}\label{Eq_Rate_1}
\rho_c = \mathbb{E} [ \rho ] = \mathbb{E} \left[ \ln \left(1+ \frac{\boldsymbol{\gamma}}{\gamma_o} \right) \right] = \int_{\gamma_{\mathrm{min}}}^{\infty} \left[ \ln \left(1+\frac{\boldsymbol{v}}{\gamma_o}\right) \right] f_{\boldsymbol{\gamma}} (v) \mathrm{d}v,
\end{equation}
%\endgroup
where the symbol $\boldsymbol{\gamma}$ is used  to represent the SINR, $v$ is merely the integration variable, while ${\boldsymbol{\gamma}_{\mathrm{min}}}$ is the effective SINR limit where no useful communication can take place below ${\boldsymbol{\gamma}_{\mathrm{min}}}$. This proposed method can capture the effect of adaptive modulation and coding schemes in communication systems. This method requires obtaining the PDF of the SINR, i.e., $f_{\boldsymbol{\gamma}}$, which can be approximated numerically from the points calculated in plotting the service success probability presented earlier, since $f_{\boldsymbol{\gamma}}=-\frac{\mathrm{d}}{\mathrm{d}v} F^c_{\boldsymbol{\gamma}}(v)$, where $F^c_{\boldsymbol{\gamma}}$ is the complementary cumulative distribution function (CCDF) of the SINR as plotted in Figs. \ref{Fig_PathLoss_Scenario}-\ref{Fig_Rician_Rayleigh_Scenario}. By means of trapezoidal integration, the resulting rate is presented in Fig. \ref{Fig_Cellular_Rate}, showing the effect of ${\boldsymbol{\gamma}_{\mathrm{min}}}$ on the network performance, where low influence of $\gamma_\mathrm{min}$ can be noted at the high extent of the $\gamma_\mathrm{min}$-axis.

\begin{figure}
 \centering
  \includegraphics[width=\linewidth]{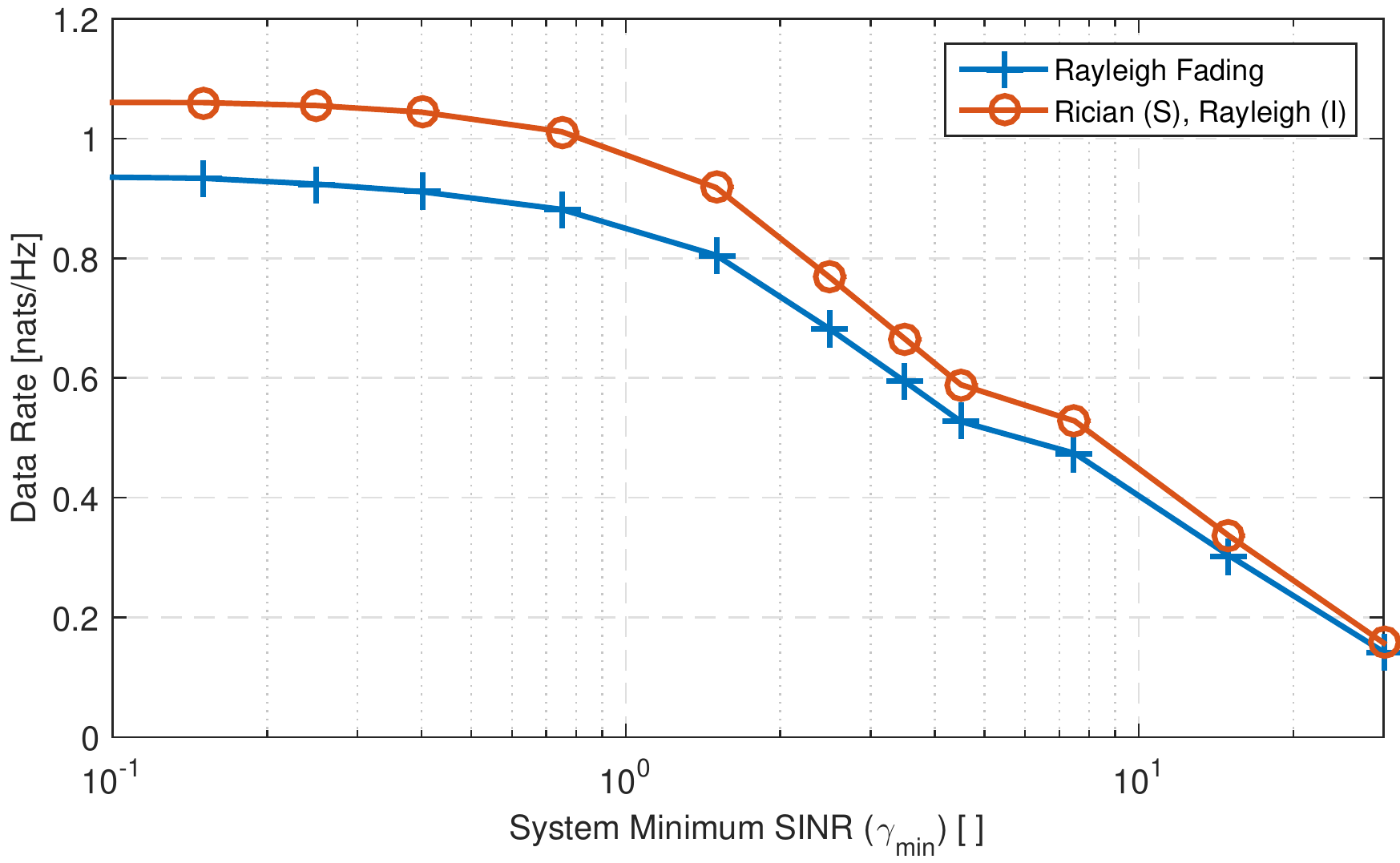}
  \caption{The expected network rate obtained using trapezoidal integration of equation (\ref{Eq_Rate_1}), versus system's minimum SINR $\gamma_\mathrm{min}$ and $\gamma_o=1$. (S) refers to the serving signal and (I) refers to the interfering signals.}\label{Fig_Cellular_Rate}
\end{figure}

Alternatively, at the cost of additional computational efforts we can utilize the following theorem to calculate the expected rate directly from the system's parameters:

\begin{Theorem}
The expected rate of a cellular network can be calculated from the following expression:
\begin{equation}\label{Eq_Capacity}
\rho_c  = \mathbb{E}_{\boldsymbol{h_o,g_o}} \left[\int_{r>0} \int_{\rho_{\mathrm{min}}}^{\infty} F_{\boldsymbol{I}} \left( \frac{P_o \boldsymbol{h_o}\boldsymbol{g_o}}{\gamma_o(e^v-1)r^\alpha}-W \right)  f_{\boldsymbol{R_o}}(r) \mathrm{d}v \mathrm{d}r \right].
\end{equation}
\end{Theorem}

\begin{proof}
We need to obtain the expectancy of the rate over three stochastic processes: (i) the spatial properties of the point process, (ii) the channel fading distributions of the serving signal, and (iii) the distribution of the interfering signals. Hence
%\begingroup\makeatletter\def\f@size{7}\check@mathfonts \allowdisplaybreaks
\begin{align}
\rho_c &= \mathbb{E} \left[ \ln \left(1+ \frac{\boldsymbol{\gamma}}{\gamma_o} \right)\right] \nonumber \\
       &= \mathbb{E}_{\boldsymbol{h_o,g_o,R_o}} \left[ \mathbb{E}_{\boldsymbol{I}}  \left[   \ln \left(1+ \frac{P_o \boldsymbol{h_o g_o} \boldsymbol{R_o}^{-\alpha}}{\gamma_o (\boldsymbol{I}+W) } \right) \right]\right] \nonumber \\
       &\stackrel{(a)}{=} \mathbb{E}_{\boldsymbol{h_o,g_o,R_o}} \left[\int_{\rho_{\mathrm{min}}}^{\infty}  \mathbb{P}\left[ \ln \left(1+ \frac{P_o \boldsymbol{h_o g_o} R_o^{-\alpha}}{\gamma_o (\boldsymbol{I}+W)  } \right) > v \right]  \mathrm{d}v  \right]\nonumber \\
       &= \mathbb{E}_{\boldsymbol{h_o,g_o,R_o}} \left[\int_{\rho_{\mathrm{min}}}^{\infty}  \mathbb{P}\left[ \boldsymbol{I} <  \frac{P_o \boldsymbol{h_o}\boldsymbol{g_o}}{\gamma_o(e^v-1)R_o^\alpha}-W \right]  \mathrm{d}v  \right],
\end{align}
%\endgroup
where $\rho_\mathrm{min} = \ln \left( 1+\frac{\gamma_\mathrm{min}}{\gamma_o} \right)$ is the minimum achievable data rates, $(a)$ follows from the fact that, for a positive random variable $X$,  $\mathbb{E}[X]=\int_0^\infty \mathbb{P} (X > v ) \mathrm{d}v$, a method which was also used in \cite{Andrews_Capacity}. %The expectation from step $(a)$ onwards is over $\boldsymbol{R_o}$, $\boldsymbol{g_o}$ and $\boldsymbol{h_o}$.
The final result of the theorem follows from averaging over $\boldsymbol{R_o}$.
\end{proof}

\section{Simulation Results} \label{Sec_Simulation}
In order to verify the \emph{coverage equation} presented in Section \ref{Sec_Coverage}, we perform Monte-Carlo simulations for more than 6,500 users distributed uniformly within a PVT cellular network. The entire simulation scenario is repeated for 80 runs, where in each run we draw a large number of BS from a Poisson distributed random variable, and deploy these BSs homogeneously over the test map, the simulation results converge very quickly and simulating more that 20 runs has no visual effect on the output plot, where the sum of the total number of simulated links exceeds 500,000. The power received from all BSs are combined at every mobile user taking into consideration the stochastic effects of the radio channel by randomly generating the individual channel gain. The SINR value is stored for each receiver, and then we obtain the resulting empirical Complementary Cumulative Distribution Function (CCDF). Noting that the CCDF of the SINR is equivalent by definition to the coverage probability, because: ${F}_{\gamma}^c(T) = 1-F_{\gamma}(T) = \mathbb{P}(\gamma > T) =p_c$.

The three different scenarios explained in Section \ref{Sec_Coverage} are simulated, namely: in scenario (1) a deterministic channel is used for all BSs, where the path-loss follows log-distance law with exponent $\alpha=4$. In scenario (2) the effect of Rayleigh channel fading is added to all BSs including the serving BS. While in scenario (3) the serving BS is assumed to favour a Rician channel with a factor $K=10$ dB. The results of the simulation runs for all three scenarios are shown in Fig. \ref{Fig_Simulation}. We observe a close match between the analytical integration of the \emph{coverage equation} from one side and Monte-Carlo simulations from another side. Also we notice how the performance seems to be bounded by the deterministic channel conditions (fading-less scenario 1) as the upper bound and the Rayleigh fading (scenario 2) as the lower bound. An important conclusion can be drawn here that the performance is strongly dependent on the serving channel fading conditions rather than the interfering signals stochastic distribution. %Furthermore, the sensitivity of the coverage towards channel fading is higher for edge users suffering a low SINR conditions.

\begin{figure}
  \centering
  \includegraphics[width=\ImagWidth]{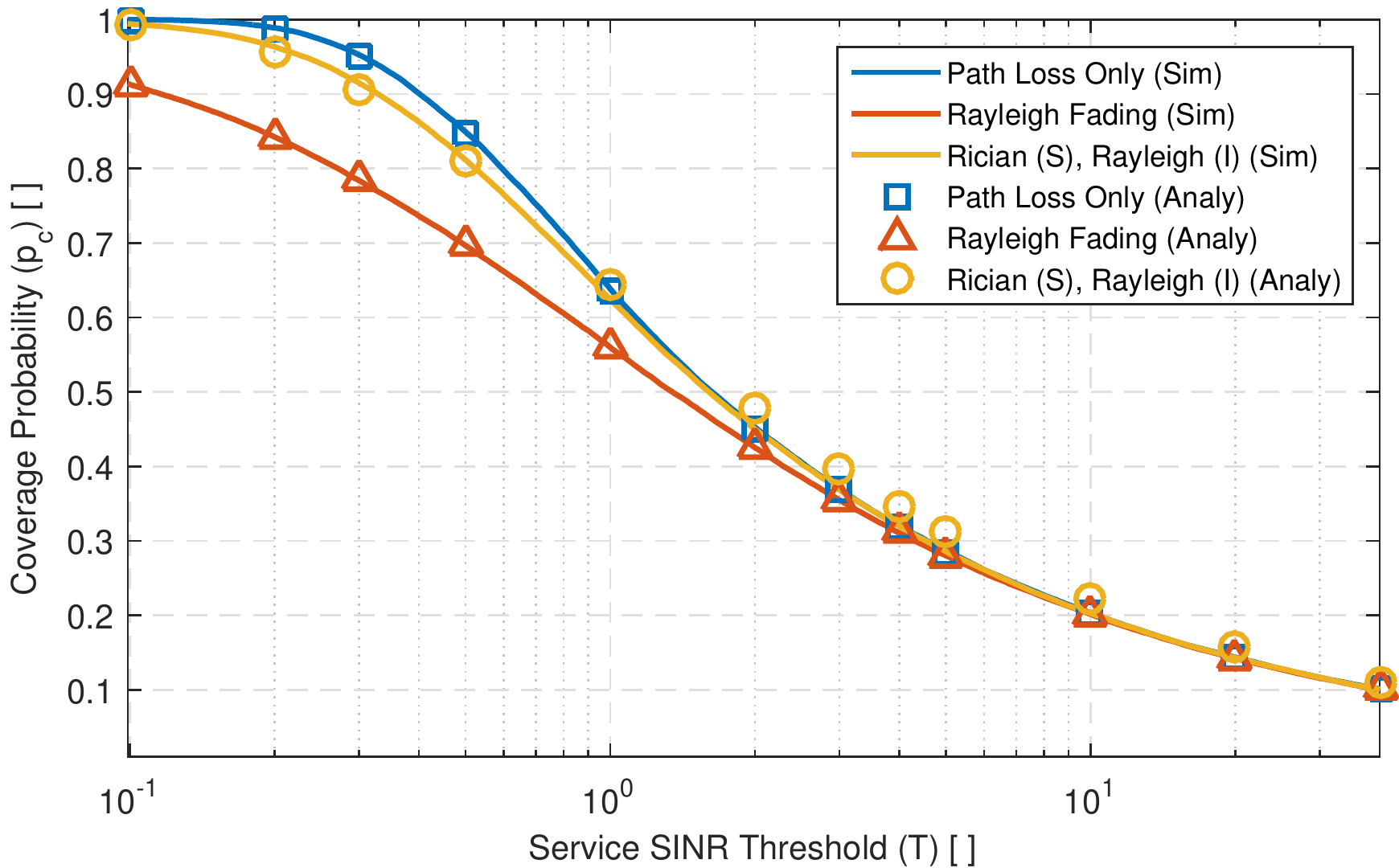}
  \caption{Coverage probability simulation results of the three proposed scenarios, compared to the analytical integration obtain using the \emph{coverage equation}. (S) refers to the serving signal and (I) refers to the interfering signals.}\label{Fig_Simulation}
\end{figure}

\section{Conclusion} \label{Sec_Conclution}
This paper has provided a mathematical framework to analytically compute the coverage and rate of random cellular networks under generic channel fading conditions. Two main observations have been made. Firstly, the stochastic process of the radio channel largely affects the coverage performance when considering lower SINR thresholds (e.g.,  the performance of cell edge users). Secondly, the density of base stations does not affect the cellular coverage when the network is interference limited, regardless of the stochastic process of the channel. For obtaining the expected rate of a cellular network two methods have been  illustrated. The first method is based on trapezoidal integration and the second method is based on a computable integration formula. Future work will include the modelling of different cellular interference coordination schemes under this framework.

\bibliography{Modelling_Capacity_Journal}

\end{document}